\documentclass[a4paper]{article}
\pdfoutput=1
\usepackage{onecolpceurws}
\institution{}

\usepackage[english]{babel}
\usepackage[utf8]{inputenc}

\usepackage{cmbright} 
\usepackage[T1]{fontenc}
\usepackage{lmodern}

\usepackage{graphicx} 
\usepackage{wrapfig}
\graphicspath{{.}{pics/}} 

\usepackage{amsmath}
\usepackage{amssymb}
\usepackage{amsthm}
\newtheorem{theorem}{Theorem}[section]

\newtheorem{corollary}[theorem]{Corollary}

\newtheorem{lemma}[theorem]{Lemma}
\newtheorem{remark}[theorem]{Remark}
\providecommand{\keywords}[1]{\textbf{\textit{Keywords: }} #1}

\usepackage{cite}
\usepackage{xcolor}
\usepackage{hyperref}
\definecolor{darkgreen}{rgb}{0,0.4,0}
\definecolor{BrickRed}{rgb}{0.65,0.08,0}
\hypersetup{colorlinks=true,linkcolor=blue,citecolor=darkgreen,filecolor=BrickRed,urlcolor=darkgreen}

\newcommand{\PR}{\mathbb{P}}

\newcommand{\Nc}{\mathcal{N}} 

\newcommand{\Cat}{\text{Cat}}
 
\newcommand{\oeis}[1]{\text{\href{https://oeis.org/#1}{{\small \tt OEIS #1}}}}

\def\IdP{{\bf 1}_\P}
\def\Pol{\operatorname{Polyo}}

\usepackage{graphicx}
\usepackage{calc}

\begin{document}

\pagestyle{empty}

\author{Cyril Banderier\\ LIPN (UMR CNRS 7030), Universit\'e de Paris Nord, France\\\url{http://lipn.fr/~banderier/}\\ 
\includegraphics[width=3mm]{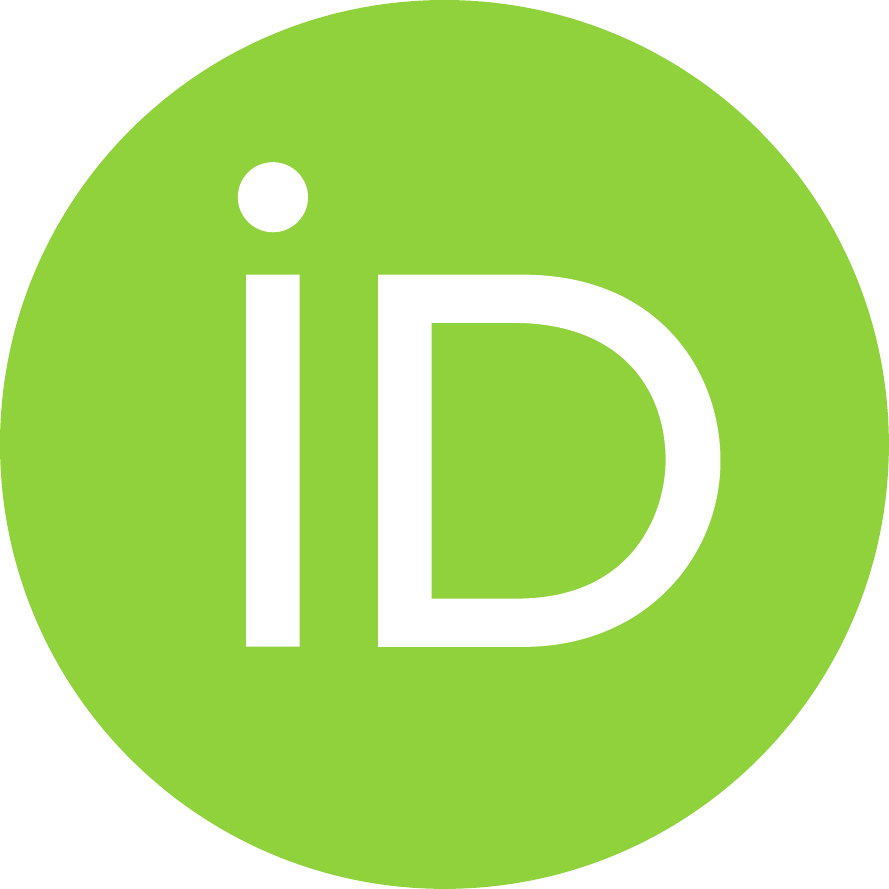} \url{https://orcid.org/0000-0003-0755-3022} 
\and \\ Philippe Marchal \\ LAGA  (UMR CNRS 7539), Universit\'e de Paris Nord, France\\\url{https://www.math.univ-paris13.fr/~marchal}\\ 
\includegraphics[width=3mm]{orcid.pdf} \url{https://orcid.org/0000-0001-8236-5713}
\and \\  Michael Wallner\\ LaBRI (UMR CNRS 5800), Universit\'e de Bordeaux, France\\\url{http://dmg.tuwien.ac.at/mwallner/} \\   
\includegraphics[width=3mm]{orcid.pdf} \url{https://orcid.org/0000-0001-8581-449X} }

\title{Rectangular Young tableaux with local decreases and the density method for uniform random generation (short version)}
\date{March 25, 2018}
\maketitle

\begin{abstract}
In this article, we consider a generalization of Young tableaux in which we allow some consecutive pairs of cells with decreasing labels.  
We show that this leads to a rich variety of combinatorial formulas,
which suggest that these new objects could be related to deeper structures, similarly to the ubiquitous Young tableaux.

Our methods rely on variants of hook-length type formulas, and also on a new efficient generic method (which we call the density method)
which allows not only to generate constrained combinatorial objects, but also to enumerate them. 
We also investigate some repercussions of this method on the  D-finiteness 
of the generating functions of combinatorial objects encoded by linear extension diagrams,
and give a limit law result for the average number of local decreases.
\end{abstract}

\keywords{Young tableau, analytic combinatorics, generating function, D-finite function, binomial numbers, random generation, density method, linear extensions of posets}

\newpage
\pagestyle{plain}
\section{Introduction}

As predicted by Anatoly Vershik in~\cite{Vershik01},  
the 21st century should see a lot of challenges and advances on the links of probability theory  
with (algebraic) combinatorics.  
A key role is played here by Young tableaux, because of their ubiquity in representation theory~\cite{Macdonald15} and 
in algebraic combinatorics, as well as their relevance in many other different fields (see e.g.~\cite{Stanley11}).

Young tableaux are tableaux with $n$ cells labelled from $1$ to $n$, 
with the additional constraint that these labels increase among each row and each column (starting from the lower left cell).
Here we consider the following variant:
What happens if we allow  exceptionally some consecutive cells with decreasing labels?
Does this variant lead to nice formulas if these local decreases are regularly placed?
Is it related to other mathematical objects or theorems?  How to generate them?
This article gives some answers to these questions.

As illustrated in Figure~\ref{fig1},
we put a bold red edge between the cells which are allowed to be decreasing.
Therefore these two adjacent cells can have decreasing labels (like $19$ and $12$ in the top row of Figure~\ref{fig1},
or $11$ and $10$ in the untrustable Fifth column), 
or as usual increasing labels (like $13$ and $15$ in the bottom row of Figure~\ref{fig1}). We call these bold red edges ``walls''.
\vspace{-1mm}
\begin{figure}[h!]
		\begin{center}	
			\includegraphics[width=.35\textwidth]{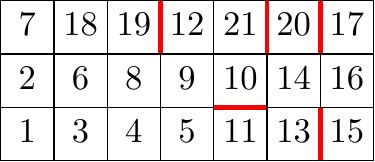}
			\caption{We consider Young tableaux in which some pairs of (horizontally or vertically) consecutive
cells are allowed to have decreasing labels. Such places where a decrease is allowed 
(but not compulsory) are drawn by a bold red edge, which we call a ``wall''.
}
			\label{fig1}
		\end{center}
\end{figure}

For Young tableaux of shape\footnote{We will refer to ``$n\times m$ Young tableaux'', or ``Young tableaux of 
shape $n\times m$'', for rectangular Young tableaux with $n$ rows and $m$ columns. They are trivially in bijection 
with $m\times n$ Young tableaux.}
 $n \times 2$
several cases lead directly to nice enumerative formulas for the total number of specific tableaux with $2n$ cells:
\begin{enumerate}
  \setlength\itemsep{0em}
	\item Walls everywhere: $(2n)!$
	\item Horizontal walls everywhere: $\frac{(2n)!}{2^n}$
	\item Horizontal walls everywhere in left (or right) column: $(2n-1)!! = \frac{(2n)!}{2^n n!}$
	\item Vertical walls everywhere: $\binom{2n}{n} = \frac{(2n)!}{(n!)^2}$
	\item No walls: $\frac{1}{n+1}{\binom{2n}{n}} = \frac{(2n)!}{(n+1) (n!)^2}$
\end{enumerate}	

In this article we are interested in the enumeration and the generation of Young tableaux (of different rectangular shapes) 
with such local decreases, and we investigate to which other mathematical notions they are related.
Section~2 focuses on the case of horizontal walls: We give a link with the Chung--Feller Theorem, binomial numbers
and a Gaussian limit law.
Section~3 focuses on the case of vertical walls: We give a link with hook-length type formulas.
Section~4 presents a generic method, which allows us to enumerate many variants of Young tableaux (or more generally, 
linear extensions of posets), and which also offers an efficient uniform random generation algorithm, and links with D-finiteness.

\pagebreak
\section{Vertical walls, Chung--Feller and binomial numbers}

\setlength{\columnsep}{20pt}%
\begin{wrapfigure}{R}{0.2\textwidth}%
	\begin{center}
		\includegraphics[width=.12\textwidth]{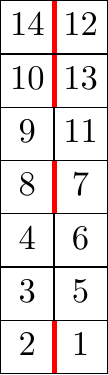}
		\caption{Example of one of our $n\times 2$ Young tableaux with walls.}
		\label{fig:2colhor}%
	\end{center}%
\end{wrapfigure}

In this section we consider a family of Young tableaux 
having some local decreases at places indicated by vertical walls, see Figure~\ref{fig:2colhor}. 

\begin{theorem}
	The number of  $n\times 2$ Young tableaux  with $k$ vertical walls is equal to 
	\begin{align*}
		v_{n,k} &= \frac{1}{n+1-k} {\binom{n}{k}} {\binom{2n}{n}}.
	\end{align*}
\end{theorem}

\begin{proof}
	We apply a bijection between two-column Young tableaux of size $2n$ with $k$ walls and Dyck paths without the positivity constraint of length $2n$ and $k$ coloured down steps. These paths start at the origin, end on the $x$-axis and are composed out of up steps $(1,1)$,  and coloured down steps $(1,-1)$ which are either red or blue. 
	
	Given an arbitrary two-column Young tableau, the $m$-th step of the associated path is an up step if the entry $m$ appears in the left column, while the $m$-th step is a down step, if the $m$-th entry appears in the right column. Furthermore, we associate colours to the down steps: If the $m$-th down step is in a row with a wall we colour it red, and blue otherwise. 
	
	Thus, $v_{n,k}$ counts the number of paths with exactly $k$ red down steps. Note that the down steps of a path below the $x$-axis are always red because a wall has to be involved, yet above the $x$-axis down steps can have any colour. 
	We decompose paths with $k$ coloured down steps with respect to the number of steps which are below the $x$-axis. By the Chung--Feller Theorem~\cite{ChungFeller49} (see also~\cite{Chen08} for a bijective proof) the number of Dyck paths of length $2n$ with $i$ down steps below the $x$-axis is independent of $i$ and equal to the Catalan number $\Cat_n = \frac{1}{n+1}{\binom{2n}{n}}$. When $i$ steps are below the $x$-axis we have to colour $k-i$ of the remaining $n-i$ steps above the $x$-axis red. This gives
	\begin{align*}
		v_{n,k} &= \sum_{i=0}^k  \binom{n-i}{k-i} \Cat_n=	\binom{n+1}{k} \Cat_n ,
	\end{align*}		
	and the claim follows.
\end{proof}

As a simple consequence, we get the following result.

\begin{corollary}
The average number of linear extensions of a random $n\times 2$ Young tableau with $k$ walls, where the
location of these walls is chosen uniformly at random, is
\begin{align*}
	\frac{1}{n+1-k} {\binom{2n}{n}}.
\end{align*}
\end{corollary}

\begin{proof}
	In a two-column Young tableau of size $2n$ we have $\binom{n}{k}$ possibilities to add $k$ walls.
\end{proof}

We now conclude this section with a limit law result.
\begin{theorem}
Let $X_n$ be the random variable for the number of walls in a random $n \times 2$ Young tableau chosen uniformly at random. 
	The rescaled random variable $\frac{X_n-n/2}{\sqrt{n/4}}$ converges to the standard normal distribution $\Nc(0,1)$.
\end{theorem}

\begin{proof}
	We see that the total number of two-column Young tableaux of size $n$ with walls is equal to 
	\begin{align*}
		\sum_{k=0}^n v_{n,k} &= \Cat_n \left(2^{n+1}-1\right).
	\end{align*}
	Then, the previous results show that
	\begin{align*}
		\PR\left(X_n = k\right) &= \binom{n+1}{k} \frac{1}{2^{n+1}-1},
	\end{align*}
	which is a slight variation of a binomial distribution with parameters $n+1$ and probability $1/2$. 
By the well-known convergence of the rescaled binomial distribution to a normal distribution the claim holds (see e.g.~\cite{FlajoletSedgewick09}).
\end{proof}

\section{Horizontal walls  and the hook-length formula}
The hook-length formula is a well-known  formula to enumerate standard Young tableaux of a given shape (see e.g.~\cite{Macdonald15,Stanley11}).
What happens if we add walls in these tableaux?
Let us first consider the case of a Young tableau of size $n$ such 
that its walls cut the corresponding tableau into $m$ disconnected parts without walls of 
size $k_1, \dots, k_m$  (e.g., some walls form  a full horizontal or vertical line).
Then, the number of fillings of such a tableau is trivially:
$$\frac{n!}{k_1! \dots k_m!} \prod_{i=1}^m \text{HookLengthFormula(subtableau of size $k_i$)}.$$
So in the rest of article, we focus on walls which 
are not trivially splitting the problem into subproblems: They are the only cases for which the enumeration (or the random generation)
is indeed challenging.

We continue our study with families of Young tableaux of shape $m \times n$
 having some local decreases at places indicated by horizontal walls in the left column.
We will need the following lemma counting special fillings of Young tableaux.
\begin{lemma}
	\label{lem:horwalls}
	The number of $n \times 2$ ``Young tableaux'' with $2\lambda$ cells filled with the numbers $1,2,\ldots, 2n$ for $n \geq \lambda$ such that (the number $2n$ is used and) all consecutive numbers between the minimum of the second column and $2n$ are used is equal to
	\begin{align}
		\label{eq:lemhorwalls}
		\binom{2n}{\lambda} - \binom{2n}{\lambda-1}.
	\end{align}
\end{lemma}

\begin{proof}
	The constraint on the maximum implies that all not used numbers are smaller than the number in the bottom right cell. Therefore it is legitimate to add these numbers to the tableaux. 
In particular, we create a standard Young tableau of shape $(\lambda,2n-\lambda)$ (i.e., the first column has $\lambda$ cells and the second one $2n-\lambda$) which is in bijection with the previous tableau. 
	
	Next we build a bijection between standard Young tableaux of shape $(\lambda,2n-\lambda)$ and Dyck paths with up steps $(1,1)$ and down steps $(1,-1)$ starting at $(0,2(n-\lambda))$, always staying above the $x$-axis and ending on the $x$-axis after $2n$ steps. In particular, if the number $i$ appears in the left column, the $i$-th step is an up step, and if it appears in the right column, the $i$-th step is a down step. 
	
	Finally, note that these paths can be counted using the reflection principle~\cite{Andre87}. In particular, there are $\binom{2n}{\lambda}$ possible paths from $(0,2(n-\lambda))$ to $(2n,0)$. 
Yet, $\binom{2n}{\lambda-1}$ ``bad'' paths cross the $x$-axis at some point. This can be seen, by cutting such a path at the first time it reaches altitude $-1$. The remaining path is reflecting along the horizontal line $y=-1$ giving a path ending at $(2n,-2)$. It is easy to see that this is a bijection between bad paths from $(0,2(n-\lambda))$ to $(2n,0)$ and all paths from $(0,2(n-\lambda))$ to $(2n,-2)$. The latter is obviously counted by $\binom{2n}{\lambda-1}$, as $\lambda-1$ of the $2n$ steps have to be up steps.
\end{proof}

\begin{theorem}
	The number of $n \times 2$ Young tableaux of size $2n$ with $k$ walls in the first column at heights $0<h_i<n$, $i=1,\ldots,k$ with $h_i < h_{i+1}$ is equal to
	\begin{align*}
		\frac{1}{2n+1} \prod_{i=1}^{k+1} \binom{2 h_i + 1}{h_i - h_{i-1}},
	\end{align*}
	with $h_0 := 0$ and $h_{k+1} := n$.
\end{theorem}

\begin{remark}
	Denoting consecutive relative distances of the walls by $\lambda_i := h_{i} - h_{i-1}$ for $i=1,\ldots,k+1$ the previous result can also be stated as
	\begin{align*}
		\frac{1}{2n+1} \prod_{i=1}^{k+1} \binom{2 (\lambda_1 + \ldots + \lambda_i) + 1}{\lambda_i}.
	\end{align*}
\end{remark}

\begin{proof}
	We will show this result by induction on the number of walls $k$. For $k=0$ the result is clear as we are counting two-column standard Young tableaux which are counted by Catalan numbers (for a proof see also Lemma~\ref{lem:horwalls} with $\lambda=n$).
	
	Next, assume the formula has been shown for $k-1$ walls and arbitrary $n$.	
	Choose a proper filling with $k$ walls and cut the tableau at the last wall at height $h_{k}$ into two parts. The top part is a Young tableau with $2(n - h_{k})$ elements and no walls, yet labels between $1$ and $2n$. Furthermore, it has the constraint that all numbers larger than the element in the bottom right cell have to be present. This is due to the fact that all elements in lower cells must be smaller. In other words, these are the objects of Lemma~\ref{lem:horwalls} and counted by~\eqref{eq:lemhorwalls}. 
	
	The bottom part is a Young tableau with $k-1$ walls and $2 h_{k}$ elements (after proper relabelling). By our induction hypothesis this number is equal to
	\begin{align*}
		\frac{1}{2h_{k}+1} \prod_{i=1}^{k} \binom{2 h_i + 1}{h_i - h_{i-1}}.
	\end{align*}
	As a final step, we rewrite Formula~\eqref{eq:lemhorwalls} into
	\begin{align*}
		\frac{2(n-\lambda)+1}{2n+1} {\binom{2n+1}{\lambda}},
	\end{align*}
	and set $\lambda := n-h_k$. Multiplying the last two expressions then shows the claim.
\end{proof}

\begin{remark}
	Note that so far we have not found a direct combinatorial interpretation of this formula.
	However, note that in general $\binom{2n+1}{\lambda}$ does not have to be divisible by $2n+1$.
\end{remark}

Let us now also give the general formula for $n \times m$ Young tableaux with walls of lengths $m-1$ from columns $1$ to $m-1$, i.e., a hole in column $m$ and nowhere else in a row with walls. Before we state the result, let us define for integers $n,k$ the falling factorial $(n)_k := n (n-1) \cdots (n-k+1)$ and for integers $n,m_1,\ldots,m_k$  such that $n \geq m_1 + \cdots + m_k$ 
the (shortened) multinomial coefficient\footnote{In the literature, one more often finds the notation 
$\binom{n}{m_1,m_2,\ldots,m_k, n-m_1-\ldots -m_k} := \frac{n!}{m_1!m_2!\cdots m_k! (n-m_1-\ldots -m_k)!}$. 
But we opted in this article for a more suitable notation to the eyes of our readers!}
$\binom{n}{m_1,m_2,\ldots,m_k} := \frac{n!}{m_1!m_2!\cdots m_k! (n-m_1-\ldots -m_k)!}$. 

\begin{theorem}
	The number of $n \times m$ Young tableaux of size $mn$ with $k$ walls from column $1$ to $m-1$ at heights $0<h_i<n$, $i=1,\ldots,k$ with $h_i < h_{i+1}$ is equal to
	\begin{align*}
		\frac{(m-1)!}{(mn+m-1)_{m-1}} 
			\left( \prod_{i=1}^{k+1} \prod_{j=1}^{m-2} \binom{\lambda_i + j}{j }^{-1} \right) 
			\left( \prod_{i=1}^{k+1} \binom{m (\lambda_1 + \ldots \lambda_i) + m-1}{\lambda_i, \ldots,\lambda_i } \right),
	\end{align*}
	where $\lambda_i := h_i - h_{i-1}$ and the $\lambda_i$'s in the multinomial coefficients appear $m-1$ times.
\end{theorem}

\begin{proof}[Proof (Sketch).]
	First derive an extension of Lemma~\ref{lem:horwalls} proved by the hook-length formula and then compute the product. Note that this gives a telescoping factor, giving the first factor.
\end{proof}

Just as one more example, here is a more explicit example of what it gives.
\begin{corollary}
	The number of $n \times 4$ Young tableaux with $k$ walls from column $1$ to $3$ at heights $0<h_i<n$, $i=1,\ldots,k$ with $h_i < h_{i+1}$ is equal to
	\begin{align*}
		\frac{6}{(4n+3)(4n+2)(4n+1)} 
			\left( \prod_{i=1}^{k+1} \frac{2}{(\lambda_i + 1)^2 (\lambda_i+2)} \right) 
			\left( \prod_{i=1}^{k+1} \binom{4 (\lambda_1 + \ldots \lambda_i) + 3 }{\lambda_i, \lambda_i, \lambda_i } \right),
	\end{align*}
	with $\lambda_i := h_i - h_{i-1}$.
\end{corollary}

Let us consider some other special cases. For example, consider tableaux with walls between every row and a hole in the last column. For this case we set $\lambda_i =1$ for all $i$. This gives the general formula
	$\frac{(mn)!}{n! (m!)^n}$,
for $n \times m$ tableaux, see \oeis{A001147} for $m=2$ and \oeis{A025035} to \oeis{A025042} for the special cases $m=3,\ldots,10$.

Now that we gave several examples of closed-form formulas enumerating some families 
of Young tableaux with local decreases, we go to harder families 
which do not necessarily lead to a closed-form result. 
However, we shall see that we have a generic method to get 
useful alternative formulas (based on recurrences), also leading to an efficient 
uniform random generation algorithm.

\pagebreak
\section{The density method, D-finiteness, random generation}

\def\f{p}
\def\Ln{f_n}
\def\P{{\mathcal P}}

In this section, we present a generic approach which allows us to enumerate and generate
any shape involving some walls located at periodic positions. To keep it readable, 
we illustrate it with a specific example (without loss of generality). 

So, we now illustrate the method on the case of a $2n \times 3$ tableau where we put walls on the right and on the left column at height
$2k$ (for  $1\leq k\leq n-1$), see the leftmost tableau in Figure~\ref{fig3}. 
In order to have an easier description of the algorithm (and more compact formulas), 
we generate/enumerate first similar tableaux with an additional cell at the bottom of  the middle column, 
see the middle tableau in Figure~\ref{fig3}: It is a polyomino $\Pol_n$ with $6n+1$ cells.
There are trivially $(6n+1)!$ fillings of this polyomino with the numbers $1$ to $6n+1$.
Some of these fillings are additionally satisfying the classical constraints of Young tableaux
(i.e., the labels are increasing in each row and each column), with some local decreases allowed between cells separated by a wall 
(as shown with bold red edges in Figure~\ref{fig3}). Let $\Ln$ be the number of such constrained fillings. 

To compute $\Ln$ we use a generic method which we call the {\em density method}, which 
we introduced and used in~\cite{Marchal14,Marchal16,BanderierMarchalWallner18}. 
It relies on a geometric point of view of the problem: 
Consider the hypercube $[0,1]^{6n+1}$  and associate to each coordinate
a cell of $\Pol_n$. To almost every element $\alpha$ of  $[0,1]^{6n+1}$ (more precisely, every element with all coordinates distinct) we can associate a filling of $\Pol_n$: Put $1$ into the cell of $\Pol_n$ corresponding
to the smallest coordinate of $\alpha$, $2$  into the cell of $\Pol_n$ corresponding
to the second smallest coordinate of $\alpha$ and so on. The reverse operation 
associates to every filling of  $\Pol_n$ a region of  $[0,1]^{6n+1}$ (which is actually a polytope). 
We call  $\P$ the set of all polytopes corresponding to correct fillings of $\Pol_n$ (i.e., respecting the order constraints).
This $\P$ is also known as the  ``order polytope'' in poset theory.

Let us explain how the density method works.
It requires two more ingredients. 
The first one is illustrated in Figure~\ref{fig3}: It is a generic building block with $7$ cells with names 
X,Y,Z,R,S,V,W. We put into each of these cells a number from $[0,1]$, which we call $x,y,z,r,s,v,w$, respectively. 
The second ingredient is the sequence of polynomials $\f_n(x)$, defined by the following recurrence  
(which in fact encodes the full structure of the problem, building block after building block):
\begin{equation} \label{rec}
\f_{n+1}(z)=\int_0^z \int_x^z  \int_0^y \int_r^z  \int_z^1 \int_y^w \f_n(v) \, dv \,dw\, ds\, dr\, dy\, dx, \text{\qquad with $\f_0=1$.}
\end{equation}

The fact that this sequence of nested integrals encodes the full structure of the problem (i.e.~all the inequalities) is  better stressed with the following writing:
\begin{equation} \label{recbis}
\f_{n+1}(z)=\int_{0<x<z} \, \int_{x<y<z} \, \int_{0<r<y} \,\int_{r<s<z} \, \int_{z<w<1} \,\int_{y<v<w}  \f_n(v) \, dv \,dw\, ds\, dr\, dy\, dx, \text{\qquad with $\f_0=1$.}
\end{equation}

\begin{figure}[h]
	\begin{center}
\def\w{.14}
\begin{tabular}{cc}
\begin{tabular}{c}
		\includegraphics[width=\w\textwidth]{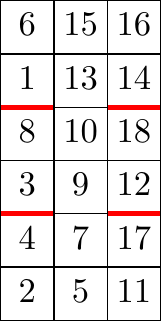} 
\end{tabular}
& \qquad
\begin{tabular}{c}
\includegraphics[width=\w\textwidth]{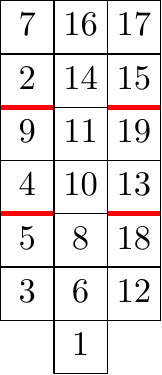}
\end{tabular}
\qquad
\begin{tabular}{c}
		\includegraphics[width=\w\textwidth]{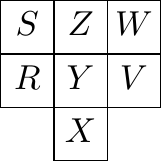}
\end{tabular}
\end{tabular}
\caption{{\bf Left:} A $2n\times 3$ Young tableau with walls. 
{\bf Centre:} Our algorithm first generates a related labelled shape, $\Pol_n$, with one more cell in its bottom
(removing this cell and relabelling the remaining cells gives the left tableau).
{\bf Right:} The ``building block'' of 7 cells. Each polyomino $\Pol_n$ is made of the overlapping of $n$ such building blocks.}
		\label{fig3}%
	\end{center}%
\end{figure}

Let us now give a more algorithmic presentation of our method:

\bigskip

\fbox{\begin{minipage}{15cm}
{\bf Density method algorithm}
\begin{itemize}
\item[1] Initialization:
We order the building blocks from $k=n-1$ (the top one) 
to $k=0$ (the bottom one). We start with the value $k:=n-1$, i.e.~the building block from the top.
Put into its cell $Z$ a random number $z$ with density $\f_n(z)/\int_0^1 \f_n(t)\, dt$.
We repeat the following process until $k=0$:
\item[2] Filling: 
Now that $Z$ is known, put into the cells $X,Y,R,S,V,W$ random numbers $x,y,r,s,v,w$ with conditional density 
$$g_{k,z}(x,y,r,s,v,w):=\frac{1}{\f_{k+1}(z)} \f_k(x){\IdP}_z,
$$
where ${\IdP}_z$ is the indicator function of the $k$-th building block (with value $z$ in cell Z):
$${\IdP}_z:={\bf 1}_{\{0\leq x\leq y \leq z, 0\leq r\leq y, r\leq s\leq z, z\leq w \leq 1, y\leq v\leq w\}}.$$
\item[3] Iteration: Consider $X$ as a the $Z$ of the next building block. Set $k:=k-1$ and go to step 2.
\end{itemize}
\end{minipage}}

\bigskip 

\noindent Next we prove that this algorithm generates Young tableaux with walls uniformly and determine its cost. 

\begin{theorem}
The density method algorithm is a uniform random generation algorithm
with quadratic time complexity and linear space complexity.
\end{theorem}
\begin{proof}
Let us indeed prove that the algorithm gives a random element of our set of polytopes $\P$ with the uniform measure. 
Our algorithm yields a $(6n+1)$-tuple ${\bold x}:=(x_j,y_i,r_i,s_i,v_i,w_i, 0\leq j\leq n,0\leq i\leq n-1)$ whose density 
is the product of the conditional densities:
\begin{equation}
\frac{\f_n(x_{n})}{\int_0^1 \f_n(t)dt}\prod_{i=1}^n g_{n-i,x_{n-i+1}}(x_{n-i},y_{n-i},r_{n-i},s_{n-i},v_{n-i},w_{n-i})
\end{equation}
The crucial point is that this product is telescopic and equal to
\begin{equation}
\frac{\f_n(x_{n})}{\int_0^1 \f_n(t)dt}\prod_{k=0}^{n-1}\frac{\f_{k}(x_{k}){\IdP}_{x_k}}{\f_{k+1}(x_{k+1})}
=\frac{p_0(x_0) \IdP}{\int_0^1 \f_n(t)\, dt}  =\frac{\IdP}{\int_0^1 \f_n(t)\, dt}  \text{\qquad (as $p_0=1$)},
\end{equation}
where ${\IdP}_{x_k}$ is as in the algorithm above the indicator function of the $k$-th block
(where the local variables $x,y,r,s,v,w,z$ of the algorithm are now $x_k,y_k,r_k,s_k,v_k,w_k,z_k$)
and where the product $\IdP$ of these indicator functions is the indicator function of the full polytope (with $n$ blocks): 
$\IdP = \prod_{k=0}^{n-1}   {\IdP}_{x_k}$.

Therefore, this density is constant on our set $\P$ of polytopes and zero elsewhere, which is exactly what we wanted. The fact that it is a density
implies that its integral is 1, whence
\begin{equation}
\int_{[0,1]^{6n+1}} \IdP\, d {\bold x}
= \int_0^1 \f_n(t) \, dt.
\end{equation}
Now if we choose a random uniform element in $[0,1]^{6n+1}$, 
the probability that it belongs to our set $\P$ of polytopes is
\begin{equation}
\int_{[0,1]^{6n+1}} \IdP\, d {\bold x}.
\end{equation}
But due to the reasoning above, this is also the probability 
that a random uniform filling of our building block is correct (i.e., respects the order constraints).  Thus this probability is given by $ \int_0^1\f_n(t)dt/(6n+1)!$. 

This implies that $\Ln=(6n+1)!  \int_0^1 \f_n(t)dt$. 

Finally, as each step relies on the computation and the evaluation of the associated polynomial $\f_n(z)$ 
(of degree proportional to $n$), this gives a quadratic time complexity, and takes linear space.
\end{proof}

{\bf Remark 1:}
If one wants to generate many diagrams and not just one, then it is valuable 
to make a precomputation phase computing and storing all the polynomials $p_n$. 
The rest of the algorithm is the same. For each new object generated, this is saving $O(n^2)$ time, to the price of $O(n^2)$ memory.
The algorithm is globally still of quadratic time complexity (because of the evaluation at each step of $p_k(x)$,
while $p_{k+1}(z)$ was already evaluated).

\smallskip

{\bf Remark 2:} If one directly wants to generate $2n \times  3$ Young tableaux with decreases
instead of our strange polyomino shapes $\Pol_n$, then one still uses the same relation between 
$\f_{n+1}$ and $\f_n$ but $\f_0$ is not defined and $\f_1$ has a more complicated form.
Another way is to generate $\Pol_n$, and to reject all the ones not having a $1$ in the bottom cell,
then to remove this bottom cell and to relabel the remaining cells from $1$ to $6n$ (see Figure~\ref{fig3}).
This still gives a fast algorithm of $O(n^2)$ time complexity 
(the only difference being the cost of the initial algorithm which is the multiplicative constant included in the big-O).

\smallskip

Using dynamic programming or clever backtracking algorithms allows hardly 
to compute the sequence $f_n$ (the number of fillings of the diagram) for $n\geq 3$.
In the same amount of time, the density method 
allows us to compute thousands of coefficients 
via the relation $f_n=(6n+1)! \int_{0}^1 p_n(z)$, where the polynomial $p_n(z)$ is computed via the recurrence
\vspace{-2mm}
\begin{equation}\label{example1}
p_{n+1}(z) = \int_{0}^z \frac{1}{24} (z-1)(x-z) (3x^3-7x^2 z-x z^2-z^3-2 x^2+4 x z+4 z^2) p_n(x) \, dx.
\end{equation}

This gives the sequence $\{f_n\}_{n\geq 0}$:\\
{\small \{1, 12, 8550, 39235950, 629738299350, 26095645151941500, 2323497950101372223250, 392833430654718548673344250, 
 115375222087417545717234273063750, 55038140590519890608190921051205837500, 40460077456664688766902540022810130044068750,
4393840235884118464495128448703896167747914784375, \dots\}.
}
\newline As far as we know, there is no further simple expression for this sequence. This concludes our analysis of the model 
given by Figure~\ref{fig3}.

\smallskip
We can additionally mention that the generating function associated to the sequence of polynomials $p_n(x)$ 
has a striking property:
\begin{theorem}
The generating function $G(t,z):= \sum_{n\geq 0} p_n(z) t^n$
is D-finite\footnote{A function $F(z)$ is D-finite if it satisfies a linear differential equation, 
with polynomial coefficients in $z$. See e.g.~\cite{FlajoletSedgewick09} for their role in enumeration and asymptotics
of combinatorial structures.}  in $z$. 
\end{theorem}
\begin{proof}
The general scheme  (whenever one has one hole between the walls) is
\begin{equation}\label{rec2}
p_{n+1}(z) = \int_{0}^z Q(x,z) p_n(x) \, dx,
\end{equation}
where $Q$ is a polynomial in $x$ and $z$, given by $Q(x,z):=\int_{\P_z} 1$. 
The fact that there is just one hole between the walls guarantees that all the other 
variables encoding the faces of the polytope $\P_z$ will disappear in this integration.
Let $d$ be the degree of $Q$  in $z$,  
applying $\frac{\partial^{d+1}}{\partial z^{d+1}}$  to both sides of Formula~\ref{rec2} leads to a relation 
between the $(d+1)$-st derivative of $p_{n+1}$ and the first $(d+1)$ derivatives of $p_n$. 
Multiplying this new relation by $t^{n+1}$ and summing over $n\geq 0$ leads to the D-finite equation for $G(t,z)$.
\end{proof}

Note that $G(t,z)$ is D-finite in $z$, but is (in general) not D-finite in $t$.
When it is D-finite in $t$, our algorithm has a better complexity (namely, a $O(n^{3/2})$ time complexity), 
because it is then possible to evaluate $p_n(z)$ in time $O(\sqrt n  \ln n)$ instead of $O(n)$.
See~\cite[Chapter 15]{BostanChyzakGiustiLebretonLecerfSalvyEtAl17} for more details on these complexity issues.

\section{Conclusion}
We presented a new way to enumerate and generate
Young tableaux with local decreases (and, more generally, linear extensions of posets).
Our approach is different from the classical way to generate Young tableaux
(e.g.~via the Greene--Nijenhuis--Wilf algorithm, see~\cite{GreeneNijenhuisWilf84}), 
which relies on the existence of an enumeration by a simple product formula (given by the hook-length formula). 
As there is no such simple product formula for the more general cases we considered here, 
 such an approach cannot work anymore.
Obviously, in order to generate these objects, there is the alternative to use some  naive ``brute-force-like'' methods 
(like e.g.~dynamic programming with backtracking). However this leads to an exponential time algorithm.
The density method which we presented here is the only method we are aware of 
which leads to a quadratic cost uniform random generation algorithm.

It would be a full project to examine many more families 
 of Young tableaux with local decreases, to check which ones lead to nice generating functions, 
to give bijections, and so on. This article presented three different approaches 
to handle them: bijections, hook-length-like formulas, and the density method. 
Let us emphasize again that the last one is of great generality.
We will give more examples in the long version of this article.

\bigskip
{\bf Acknowledgments}: 
This work was initiated during the postdoctoral position 
of Michael Wallner at the University of Paris Nord, in September-December 2017, thanks a MathStic funding.
The subsequent part of this collaboration was funded by the Erwin Schr{\"o}dinger Fellowship
of the Austrian Science Fund (FWF):~J~4162-N35.

\bibliographystyle{alpha} 
\newcommand{\etalchar}[1]{$^{#1}$}

\end{document}